\documentclass{llncs}
\usepackage{url}
\usepackage{hyperref}
\usepackage{amsmath}
\usepackage{amssymb}
\usepackage{stmaryrd}

\usepackage{paralist}

\usepackage{graphicx}
\usepackage{todonotes}
\title{Towards a constraint solver for proving confluence with invariant and equivalence of realistic CHR programs\thanks{This work is supported by The Danish Council for Independent Research, Natural Sciences, grant no.~DFF 4181-00442.}}
\author{Henning Christiansen \and Maja H.\ Kirkeby}
\authorrunning{M.~H.~Kirkeby \and H.~Christiansen}
\institute{Computer Science, Roskilde University, Denmark\\
\email{majaht@ruc.dk} and \email{henning@ruc.dk}}

\usepackage{stackengine}

\usepackage{tikz}
\usetikzlibrary{positioning,shapes,fit,arrows.meta,decorations}
\usetikzlibrary{shapes.geometric}
\usetikzlibrary{decorations.pathreplacing,decorations.pathmorphing}

\newcommand\ttleq{\raisebox{1pt}{$\mathrel{\stackunder[-1.5pt]{\texttt{<}}{\texttt{-}}}$}}

\def\true{\mathit{true}}

\def\failure{\mathit{failure}}
\def\error{\mathit{error}}
\def\exe{\mathit{Exe}}

\def\where{\mathrel{\mbox{\textsc{where}}}}

\def\capplus{\mathop{\setbox1=\hbox{$\cap$}\vbox to \ht1{\hbox{$\cap$}\vss\hbox to \wd1{\hfil\scriptsize$+$\hfil}}}}

\def\ourmapsto{\mapsto}
\def\ourmapsfrom{\mapsfrom}

\def\ourlongmapsto{\mathrel{\mbox{\hbox to 0.2ex{}\vrule width 0.15ex height 1ex\hskip -0.2ex$\longrightarrow$}}}

\def\mapstostar{\stackrel{*}\ourmapsto}
\def\mapsfromstar{\stackrel{*}\ourmapsfrom}

\def\supmeta{^{\mathit{meta}}}

\def\sublogic{_{\mathit{logic}}}
\def\subprolog{}
\def\Isubprolog{I}

\def\rightarrowsubprolog{\mathrel{\ourmapsto}}
\def\leftarrowsubprolog{\mathrel{\ourmapsfrom}}

\def\approxprolog{\sim}

\def\approxmetaprolog{\approx}

\def\Isubmetaprolog{\mathord{I\!\!I}}

\def\Mmetaprolog{\mathcal{M}}

\def\Iprolog{I}
\def\Imetaprolog{\Isubmetaprolog}

\def\cinv{\text{{\rm{\texttt{inv}}}}}
\def\cequiv{\text{{\rm{\texttt{equiv}}}}}
\def\cfreshvars{\text{{\rm{\texttt{freshvars}}}}}

\def\denotes#1{[\![#1]\!]}
\def\denotesgr#1{[\![#1]\!]^{\mathit{Gr}}}

\def\denotesmetaprolog#1{[\![#1]\!]}

\def\ourlongmapstometalogic{\mathrel{\mbox{\hbox to 0.2ex{}\vrule width 0.15ex height 1ex\hskip -0.2ex$\longrightarrow$}\sublogic\supmeta}}
\def\ourlongmapsfrommetalogic{\mathrel{\mbox{\hbox to 0.2ex{}$\longleftarrow$\hskip -0.3ex\vrule width 0.15ex height 1ex}\sublogic\supmeta}}

\def\ourlongmapstometaprolog{\Longmapsto}
\def\ourlongmapsfrommetaprolog{\Longmapsfrom}

\def\eyes{\hbox to 1.84ex{\vbox to 0pt{\vss\hbox to 1.84ex{\hfil\vbox to 0pt{\hbox{\textrm{.}}}\hfil\vbox to 0pt{\hbox{\textrm{.}}\vss}\hfil}}}}
\def\ssmile{\hbox to 1.84ex{\vbox to 0pt{\vss\hbox to 1.84ex{\hfil\vbox to 0pt{\vss\hbox{\tiny$\smile$}\vss}\hfil}}}}
\def\ffrown{\hbox to 1.84ex{\vbox to 0pt{\vss\hbox to 1.84ex{\hfil\vbox to 0pt{\vss\hbox{\tiny$\frown$}\vss}\hfil}}}}
\def\bvadr{\hbox to 1.84ex{\vbox to 0pt{\vss\hbox to 1.84ex{\hfil\vbox to 0pt{\vss\hbox{\tiny$\leftrightsquigarrow$}\vss}\hfil}}}}

\def\willsucced{\hbox{\hskip 0.161ex\vrule height 1.84ex\vbox to 1.84ex
        {\hrule width 1.84ex\vskip 0.391ex\eyes\vskip -1.886ex\ssmile\vss\hrule width 1.84ex}\vrule height 1.84ex}\hskip 0.18ex}
\def\willfail{\hbox{\hskip 0.161ex\vrule height 1.84ex\vbox to 1.84ex
        {\hrule width 1.84ex\vskip 0.391ex\eyes\vskip -1.886ex\ffrown\vss\hrule width 1.84ex}\vrule height 1.84ex}\hskip 0.18ex}
\def\willerror{\hbox{\hskip 0.161ex\vrule height 1.84ex\vbox to 1.84ex
        {\hrule width 1.84ex\vskip 0.391ex\eyes\vskip -1.886ex\bvadr\vss\hrule width 1.84ex}\vrule height 1.84ex}\hskip 0.18ex}

\def\namedrule#1#2#3#4#5{#1\colon #2 \text{\rm\texttt{\char92}}#3\, \text{\rm\texttt{<=>}} \,#4 \text{\rm\texttt{|}} #5}
\def\namedrulenoguard#1#2#3#4{#1\colon #2 \text{\rm\texttt{\char92}}#3\, \text{\rm\texttt{<=>}} #4}
\def\nonamerule#1#2#3#4{#1 \text{\rm\texttt{\char92}}#2\, \text{\rm\texttt{<=>}} \,#3 \text{\rm\texttt{|}} #4}
\def\namedsimprule#1#2#3#4{#1\colon #2  \text{\rm\texttt{<=>}} \,#3 \text{\rm\texttt{|}} #4}
\def\nonamesimprule#1#2#3{#1  \text{\rm\texttt{<=>}} \,#2 \text{\rm\texttt{|}} #3}
\def\nonamesimprulenoguard#1#2{#1  \text{\rm\texttt{<=>}} \,#2 }

\def\nonameproprule#1#2#3{#1  \text{\rm\texttt{==>}} \,#2 \text{\rm\texttt{|}} #3}

\begin{document}
\maketitle
\begin{abstract}
Confluence of a nondeterministic program ensures a functional input-output relation, freeing the programmer from considering the actual scheduling strategy, and allowing optimized and perhaps parallel implementations. The more general property of confluence modulo equivalence ensures that equivalent inputs are related to equivalent outputs, that need not be identical. Confluence under invariants is also considered. Constraint Handling Rules (CHR) is an important example of a rewrite based logic programming language, and we aim at a mechanizable method for proving confluence modulo equivalence of terminating programs. 
While earlier approaches to confluence for CHR programs concern an idealized logic subset, we refer to a semantics compatible with standard Prolog-based implementations. 
We specify a meta-level constraint language in which invariants and equivalences can be expressed and manipulated, extending our previous theoretical results towards a practical implementation.
\end{abstract}

\section{Introduction}
Confluence of a program consisting of rewrite rules means that its input-output relation is functional, even if the
selection of internal computation steps is nondeterministic. Thus the underlying implementation is free to choose
any optimal ordering of steps, perhaps in parallel, and the programmer do not need to care about -- or even know about --
the internal scheduling strategy.

Constraint Handling Rules (CHR) is a rewrite rule-based language~\cite{DBLP:conf/iclp/Fruhwirth93,fruehwirth-98,fru-chr-book-2009}. The property of confluence of CHR programs has been studied since the introduction of CHR, e.g.,~\cite{DBLP:conf/cp/Abdennadher97,DBLP:conf/cp/AbdennadherFM96,DBLP:journals/constraints/AbdennadherFM99,HenningMajaLOPSTR18,DBLP:conf/flops/GallF18}. Most of these confluence results are developed with respect to a logic-based semantics reflecting only the logical subset of CHR; this choice gave elegant confluence proofs. More recently confluence results are developed for a more realistic semantics~\cite{DBLP:conf/lopstr/ChristiansenK14,DBLP:journals/fac/ChristiansenK17};  realistic in the sense that it reflects the de-facto standard implementations of CHR upon Prolog, e.g.,~\cite{DBLP:journals/aai/HolzbaurF00a,SchrijversDemoen2004}, including a correct treatment of Prolog's non-logical devices (e.g., \texttt{var}/1, \texttt{is}/2) and run-time errors. Here we focus on this Prolog-based semantics.

While confluence means that the input-output relation is functional, the more general property confluence modulo equivalence means that the input-output relation may be seen as a function from equivalence classes of the input to equivalence classes of the output. In this case, we again obtain the free choice of computation order. 
Typical examples include redundant data structures (e.g., representing a set as a list in which the order is immaterial) and
algorithms searching for a single best solution among several (two solutions are considered equivalent when they have the same score), e.g., the Viterbi algorithm~\cite{DBLP:journals/fac/ChristiansenK17}.
\begin{example}[\cite{HenningMajaLOPSTR18,DBLP:journals/fac/ChristiansenK17,ChristiansenKirkebyLOPSTR18}]\label{ex:set}
The following non-confluent program is a typical example of programs using redundant data structures; it collects items into a set represented as a list.
$$\begin{array}{l}
\nonamesimprulenoguard{\texttt{set(L),}\, \texttt{item(X) }}{\texttt{ set([X|L])}}\\
\end{array}
$$ 
This rule will apply repeatedly, replacing constraints matched by the left hand-side with those indicated on the right. 
For instance, the input query
\begin{verbatim}
?- set([]), item(a), item(b).\end{verbatim}
may produce two different outputs: \texttt{set([a,b])} and \texttt{set([b,a])}, both representing the same set.
\end{example}
Often programs are developed with a particular set of initial queries in mind. This includes the \texttt{set}-program of Example~\ref{ex:set}, reflecting a tacitly assumed state invariant: only one \texttt{set}-constraint is allowed. 
 While confluence requires that all states, including states that were never intended to be reached, satisfy the confluence property, confluence under invariant only 
requires that states in the invariant satisfy the confluence property.
Thus when a program is confluent under invariant, it means that a part of the input-output relation is functional, namely that which is restricted to invariant states. 
Due to the non-logical built-in constraints of the Prolog-based semantics, the empty program is not confluent, see Example~\ref{ex:empty-is-non-confl}; thus no program is confluent under this semantics. We identify a basic invariant that ensures confluence (under invariant) of the empty program. 

Most confluence results of CHR, and those we present here, are based on
Newman's lemma~\cite{Newman42} and a similar lemma for confluence modulo equivalence by Huet~\cite{DBLP:journals/jacm/Huet80} (Section~\ref{sec:prelim}), and are only relevant for terminating programs.
Methods for proving termination of CHR programs have been considered by, e.g.,~\cite{DBLP:conf/compulog/Fruhwirth99,DBLP:conf/iclp/PilozziS09a,DBLP:conf/lopstr/PilozziS11}.
An example of a non-terminating CHR program is shown in Example~\ref{ex:gcd} below, which
we can show confluent under a suitable invariant that enforces termination.
\begin{example}\label{ex:gcd}
The following program implements a version of Euclid's algorithm for computing the greatest common divisor of two (or more) positive integers.  
$$
\begin{array}{l}
\namedrulenoguard{r_1}{\texttt{gcd(N)}}{\texttt{gcd(N)}}{\texttt{ true.}}\\
\namedrule{r_2}{\texttt{gcd(N)}}{\texttt{gcd(M)}}{\texttt{ N<M, L is M-N }}{\texttt{ gcd(L).}}
\end{array}
$$
These rules will apply repeatedly; the first rule removes duplicates and the second one considers two unequal integers and replaces the larger by their distance.
For example, the input query 
\begin{verbatim} 
 ?- gcd(49), gcd(63).\end{verbatim}
produces the output \texttt{gcd(7)}. In general, the program is a nonterminating since inputs with a non-positive number may make the program loop, e.g., the second rule applies to \texttt{gcd(0),} \texttt{gcd(1)} producing the same state. In addition, the program may produce run-time errors, for instance by \texttt{gcd(1),} \texttt{gcd(1/0)} (since \texttt{<}/2 in the Prolog-based implementation/semantics evaluates its arguments). 
\end{example}
We introduce a small but useful extension to this line of work; we propose the weaker property termination under invariant\footnote{A similar property has been studied by~\cite{Endrullis2010} in term rewriting
under the name of  ``local termination''.}
(i.e., not necessarily termination outside the invariant), which is sufficient when proving confluence (modulo equivalence). 
\begin{example}\label{ex:gcd:inv}(Example~\ref{ex:gcd} continued)
For the \texttt{gcd}-program, a relevant invariant would be one restricting the states to those containing only \texttt{gcd}-constraints with positive integer arguments; the program is terminating and confluent under this invariant, and moreover it excludes states that may cause run-time errors.
\end{example}
When the invariants remove the inputs that may cause undesired program behavior such as non-termination or encountering of run-time errors, the invariant restricted input-output function is not only partial, but total.

\subsection{Related Work}
Most earlier approaches worked only for an idealized set of logical built-ins  managed by a ``magic'' constraint solver that has no counterpart in standard CHR implementations, 
e.g.,~\cite{DBLP:conf/cp/Abdennadher97,DBLP:conf/cp/AbdennadherFM96,DBLP:journals/constraints/AbdennadherFM99,HenningMajaLOPSTR18,DBLP:conf/iclp/DuckSS07}.
It was suggested by~\cite{DBLP:conf/iclp/DuckSS07} to take invariants into account (under the name of observable confluence), staying within the first-order framework.
Due to the lack of a more expressive meta-level, including a formal representation of invariants, this approach explodes (as also noticed by its authors) into infinitely many proof cases for simple invariants such as groundness.
In~\cite{HenningMajaLOPSTR18}, we have shown how this problem can be eliminated for the mentioned logic-based
semantics, introducing a meta-level representation analogously to what we do in the present paper. 

Confluence modulo equivalence was introduced  for CHR by~\cite{DBLP:conf/lopstr/ChristiansenK14},
also arguing to use a Prolog-based semantics closer to current CHR implementations. 
An in-depth theoretical analysis of these issues is given by~\cite{DBLP:journals/fac/ChristiansenK17},
also suggesting to use a ground meta-level representation for invariants and equivalences.
An attempt to handle confluence modulo equivalence was made by~\cite{DBLP:conf/flops/GallF18} based
on the approach of~\cite{DBLP:conf/iclp/DuckSS07}, inheriting the mentioned explosion problem. 

\subsection{Contributions}
In this work we focus on the more realistic Prolog-based CHR semantics that is compatible with Prolog-based CHR implementations. We present essential steps towards a constraint solver for proving confluence modulo equivalence under invariants for CHR programs that are terminating (under invariant). 
%
We extend previous approaches introducing  a specific meta-logic constraint language in which we can refer to states as data objects and better reason with invariants and equivalences. We use the meta-language for specifying abstract transitions and joinability (modulo equivalence) proofs in relation to the Prolog-based semantics. 
\subsection{Overview}
Section~\ref{sec:prelim} provides the basic definitions and properties of transitions systems,
and Section~\ref{sec:Prolog-based-CHR} recalls the syntax and operational semantics for
CHR with correct handling of non-logical built-ins such as \texttt{var}/1, \texttt{nonvar}/1 and \texttt{is}/2.
Section~\ref{sec:metalevel} introduces the meta-level transition system for CHR, including
its meta-level constraint language, and Section~\ref{sec_confluence} shows our confluence results.
Section~\ref{sec:conclude} gives some concluding remark and plans for future work.

\section{Preliminaries}\label{sec:prelim}
We give the basic definitions in a compact form here; background and motivating examples may be found in~\cite{DBLP:journals/fac/ChristiansenK17}.
A \emph{transition system} is a pair $\langle A,\ourmapsto\rangle$, where $A$ is a set of \emph{states} and
${\ourmapsto}\subseteq A\times A$ is the \emph{transition relation};
$\mapstostar$ is the reflexive transitive closure of $\ourmapsto$.
An element of $s\in A$ is \emph{final} or \emph{normal form} if $\nexists s'\colon s \ourmapsto s'$.
The system is \emph{terminating} whenever every  transition sequence $s_1\ourmapsto s_2\ourmapsto \cdots$ is finite.
An \emph{invariant} $I$ is a subset of $A$ such that $s\in I \land s\ourmapsto s'\Rightarrow s'\in I$. Whenever $s\in I$, $s$ is an \emph{$I$-state}.
A \emph{state equivalence} is an equivalence relation over $A$, typically written as $\sim$;
the  set of equivalence classes given by $\sim$ is written $A/{\sim}$.
In the context of an invariant $I$, $\ourmapsto$ and $\sim$ are tacitly assumed to be restricted to $I$-states.

Two states $s_1,s_2$ are \emph{joinable (modulo $\sim$)} whenever there exists a
state $s_{12}$ (states $s'_1, s'_2$) such that
$s_1\mapstostar s_{12}\mapsfromstar s_2$
($s_1\mapstostar s'_1\sim s'_2\mapsfromstar s_2$).

A system is \emph{confluent} whenever, for any $s_0, s_1,s_2$ with  $s_1\mapsfromstar s_0 \mapstostar s_2$, that $s_1, s_2$ are joinable.
It is \emph{confluent modulo $\sim$} whenever, for any $s_0, s_0', s_1,s_2$ with  $s_1\mapsfromstar s_0\sim s_0'\mapstostar s_2$, that $s_1, s_2$ are joinable modulo $\sim$.

An \emph{$\alpha$-corner} is of the form $s_1\ourmapsfrom s_0 \ourmapsto s_2$; a \emph{$\beta$-corner} is of the form $s_1\sim s_0 \ourmapsto s_2$;  in both cases, the indicated relationships must hold.
A system is \emph{locally confluent (modulo $\sim$)} whenever all $\alpha$- ($\alpha$- and $\beta$-) corners are joinable (joinable modulo $\sim$).
The following well-known lemmas reduce proofs of confluence (modulo equivalence) for terminating systems to proofs of the simpler property of local confluence (modulo equivalence).

\begin{lemma}[Newman \cite{Newman42}]\label{lm:newman}
A terminating transition system is confluent if and only if
it is locally confluent.
\end{lemma}

\begin{lemma}[Huet \cite{DBLP:journals/jacm/Huet80}]\label{lm:huet-mod-eq}
A terminating transition system is confluent modulo $\sim$ if and only if
it is locally confluent modulo $\sim$.
\end{lemma}
In accordance with~\cite{DBLP:conf/iclp/DuckSS07},
these extends to \emph{observable} confluence (modulo $\sim$) in the context of an invariant, i.e.,
when $\ourmapsto$ (and $\sim$)  is restricted to $I$-states. Similarly, transitions, corners, etc.\ whose states are $I$-states are called \emph{observable}.

%

We suggest the following generalization: a transition system is \emph{terminating under $I$}, whenever every observable transition sequence is finite.
Lemmas~\ref{lm:newman} and~\ref{lm:huet-mod-eq} 
 generalize immediately
for termination under invariant as follows.

\begin{lemma}
An observably terminating transition system is observably confluent (modulo $\sim$) if and only if it is locally observably confluent (modulo $\sim$).
%
\end{lemma}
\begin{proof}
Follows immediately from
Lemmas~\ref{lm:newman} and~\ref{lm:huet-mod-eq} for an $I$-restricted system. 
\end{proof}

%

\section{Prolog-based CHR}\label{sec:Prolog-based-CHR}
We recall the syntax and operational semantics of CHR with non-logical Prolog built-ins such as \texttt{var}/1, \texttt{nonvar}/1 and \texttt{is}/2.
These and any other built-ins are executed immediately  resulting
in a substitution (perhaps a failure or error substitution as defined below) that is applied to the current state,
analogous to how a typical CHR implementation upon Prolog works, e.g.,~\cite{DBLP:journals/aai/HolzbaurF00a,SchrijversDemoen2004}.

Standard first-order notions of variables, terms, predicates, atoms, etc.\ are assumed.
We extend substitutions to include special mappings between states.
A \emph{(proper) substitution} is a mapping from a finite set of variables to terms, e.g., the substitution $[x/t]$ replaces variable $x$ by term $t$. 
For a proper substitution $\sigma$  and expression $E$, $E\sigma$
(or $E\cdot\sigma$)
denotes the expression that arises when $\sigma$ is applied to $E$; composition of two proper substitutions $\sigma, \tau$ is denoted $\sigma\circ\tau$.
A \emph{substitution} $\sigma$ extends proper substitutions with two special substitutions $\failure$ and $\error$; applying special substitution $\failure$ (resp.\ $\error$) to a state $s\notin\{\error,\failure\}$ is defined as $(s)\failure=\failure$ (resp.\ $(s)\error=\error$).\footnote{Application of special substitutions to $\failure$ and $\error$ states is undefined.}
%

Two disjoint sets of \emph{constraint predicates}
are assumed, namely  \emph{user constraint} predicates and \emph{built-in constraint} predicates.
The meaning of built-in predicates is given by a function $\exe$ that maps them into substitutions. The set of built-in predicates may vary, see~\cite{DBLP:conf/lopstr/ChristiansenK14,DBLP:journals/fac/ChristiansenK17}, but we assume always \texttt{=}/2 with 
$\exe(t_1\texttt{=}t_2)$ being a most general unifier of $t_1,t_2$ if one exists, and $\failure$ otherwise.  
$\exe$ is extended to sequences of built-ins as follows, which is not commutative.
$$
\exe((b_1,b_2)) =
\begin{cases}
\exe(b_1) & \text{when $\exe(b_1)\in\{\failure,\error\}$},\\
\exe(b_2\,\exe(b_1)) & \text{when otherwise $\exe(b_2\,\exe(b_1))$}\\
& \text{~~~~~~~~~~~~~~~~~~~$\in\{\failure,\error\}$},\\
 \exe(b_1)\circ\exe(b_2\,\exe(b_1)) & \text{otherwise\label{text-def:exe}}
\end{cases}
$$
Thus the special substitutions are absorbing in the sense that once  $\failure$ or $\error$ has been detected,
this cannot be changed.

\begin{example}\label{ex:is=builtins}
We would expect $\exe(\texttt{nonvar(a)})=\emptyset$ and  $\exe(\texttt{nonvar(X)})=\failure$.
When $\exe$ produces the $\error$ substitution, it  indicates that the built-in cannot
evaluate due to insufficient instantiation or anomalies such as division by zero.
Assuming that $\exe$ models the  \texttt{is}/2 predicate  as it is defined in Prolog, we get the following.
$$
\exe((\texttt{X}\; \texttt{is}\, \texttt{Y+1}, \texttt{Y=2}))  =  \error
\qquad
\exe((\texttt{Y=2},\, \texttt{X}\; \texttt{is} \;\texttt{Y+1}))  =  \{\texttt{Y}/\texttt{2},\texttt{X}/\texttt{3}\}
$$
\end{example}

\subsection{Syntax}
We use the generalized simpagation form~\cite{fru-chr-book-2009} to capture all rules of CHR.
A \emph{rule} is a structure of the form
$\nonamerule{H_1}{H_2}GC$, where 
$H_1 \texttt{\char92}H_2$ is the \emph{head} of the rule,  where
 $H_1$ and $H_2$ are sequences, not both empty, of user constraints,
$G$ is the \emph{guard} which is a sequence of built-in constraints, 
and $C$ is the \emph{body} which is a sequence of constraints of either sort.
When $H_1$ is empty, the rule is a \emph{simplification}, which may be written
$\nonamesimprule{H_2}G C$;
when $H_2$ is empty, it is a \emph{propagation}, which may be written $\nonameproprule{H_1}G C$;
any other rule is a \emph{simpagation}; when $G=\true$,  $(G \texttt{|})$ may be left out. 
The \emph{head variables} of a rule are those appearing in the head, any other variable is \emph{local}.
A \emph{pre-application} of a rule $r = (\nonamerule{H_1}{H_2}GC)$ is of the form
$(\nonamerule{H_1'}{H_2'}{G'}{C'})\sigma$
where 
$r' = (\nonamerule{H_1'}{H_2'}{G'}{C'})$ is a variant with fresh variables of $r$
and $\sigma$ is a substitution to the head variables of $r'$, where, for no variable $x$, $x\sigma$ contains a local variable of $r'$. The notions of head variables and local variables extends naturally to pre-applications. 

\subsection{Operational Semantics}
A \emph{state representation (s.repr.)} is either a multiset of constraints (a \emph{proper} state) or one of $\failure$ or $\error$; applying the $\failure$ ($\error$) substitution to a proper s.repr.\ yields the $\failure$ ($\error$) s.repr.
A \emph{state} is an equivalence class of s.repr.s under variable renaming; $\failure$ or $\error$ identify their own classes.

A \emph{rule application} wrt.\ a proper state repr.\ $S$
is a pre-application $H_1 \texttt{\char92}H_2\;\texttt{<=>} \;G \texttt{|} C$ 
such that $\exe(G)$ is a proper substitution 
that does not instantiate the pre-application's head variables, i.e., $H_i = H_i\exe(G)$ for $i\in \{1,2\}$.

There are two sorts of transition steps, \emph{by rule application} and \emph{by built-in}.
\begin{eqnarray*}
H_1\uplus H_2\uplus S & \;\rightarrowsubprolog\; &  (H_1\uplus C\uplus S) \,\exe(G)\\
&& \qquad \text{when there exists a rule application $H_1 \texttt{\char92}H_2 \texttt{<=>} G \texttt{|} C$} \\
 \{b\}\uplus S &  \;\rightarrowsubprolog\; &   S\,\exe(b)\qquad \text{for a built-in $b$ constraint}
\end{eqnarray*}
Notice that for a built-in step, the resulting state may be $\failure$ or $\error$.

\subsection{Invariants and Confluence}
CHR programs without invariants are often not confluent. 
\begin{example}\label{ex:zigzagPlain}
We consider the CHR program consisting of the following four rules.
\begin{center}
\indent\begin{tabular}{ccl}
$r_1$:\quad\texttt{p(X)} \texttt{<=>} \texttt{q(X)} &\qquad\qquad &
   $r_3$:\quad\texttt{q(X)} \texttt{<=>} \texttt{X>0 |} \texttt{r(X)} \\
$r_2$:\quad\texttt{p(X)} \texttt{<=>} \texttt{r(X)} &\qquad\qquad &
   $r_4$:\quad\texttt{r(X)} \texttt{<=>} \texttt{X$\ttleq$0 |} \texttt{q(X)}
\end{tabular}
\end{center}
It is not confluent since the corner
$\texttt{q(X)}\,{\mapsfrom}\,\texttt{p(X)}\,{\ourmapsto}\,  \texttt{r(X)}$
is not joinable. 
However, adding the invariant ``reachable from an initial state \texttt{p($n$)} where $n$ is a number''
makes the program observably confluent, as shown in Example~\ref{ex:zigzagPlain-obscorners} below.
\end{example}
In fact, depending on which built-ins that are defined in the semantics, even the empty-program may not be confluent. 
\begin{example}\label{ex:empty-is-non-confl}
Assuming that the semantics includes Prolog-like built-ins \texttt{is}/2 and \texttt{=}/2, even the empty program is non-confluent; e.g., the corner $( \error \leftarrowsubprolog \{ \texttt{X} \; \texttt{is}\, \texttt{Y+1}, \texttt{Y=2} \} \rightarrowsubprolog \{\texttt{X} \; \texttt{is}\, \texttt{2+1}\})$ is not joinable.
\end{example}
These examples signify the need for considering observable confluence for Prolog-based CHR, rather than classical confluence.

\section{A meta-level transition system for CHR}\label{sec:metalevel}
We introduce a transition system whose individual states each cover typically infinitely many
CHR states, and such that each transition (sequence) analogously covers typically infinitely many
transition (sequence)s.
In this system, we can formulate critical corners, which, if shown joinable, implies joinability of all covered CHR corners.
Thus a proof of local confluence of a given CHR program may be reduced to a finite number of cases.
The difference between this and earlier approaches
is that we step up to a meta-level representation in which we can represent invariant and state equivalence
statements. 

We assume a specific CHR program together with an $\exe$ function
giving meaning to its built-ins plus  invariant $\Isubprolog$ and state equivalence $\approxprolog$.
By the \emph{object-level system}, we refer to the transition system thus induced by the operational semantics
for CHR given above.

Any object term, formula, etc.~is \emph{named} by a ground \emph{meta-level term}.
Variables are named by special constants, say \texttt{X} by \texttt{'X'},
and any other symbol by a function symbol written the same way;
e.g., the non-ground object level atom \texttt{p(A)} is named by the ground meta-level term \texttt{p('A')}.
For any such ground meta-level term $mt$, we indicate the object it names as $\denotesgr{mt}$. For example,
  $\denotesgr{\texttt{p('A')}}= \texttt{p(A)}$ and
  $\denotesgr{\texttt{p('A')}\land\texttt{'A'>2}} = (\texttt{p(A)}\land\texttt{A>2})$.

To describe meta-level states, we assume a 
constructor, $\langle-,-\rangle$, the \emph{state (representation) constructor}.
Intuitively, when we indicate a meta-level state as $\langle S,B\rangle$, $S$ will represent the multiset of CHR constraint in the covered states, and $B$  is not intended to represent a component of these
states, but is a sequence of built-ins to be executed and applied to $S$ (in suitable way) to form CHR states.

For a given object entity $e$, we define its \emph{lifting}\label{inline:lifting}
 to the meta-level by
1) selecting a meta-level term that names $e$, and 2) replacing variable names
in it consistently by fresh meta-level variables.
For example, $\texttt{p(X)}\land\texttt{X>2}$ is lifted to $\texttt{p($x$)}\land\texttt{$x$>2}$, where \texttt{X} and $x$
are object variable and, resp., meta-variable. We will refer to a lifted $n$-ary built-in constraint $\texttt{b}(x_1, \ldots, x_n)$ where $x_1, \ldots, x_n$ are (unique) meta-variables as a \emph{built-in template}.
By virtue of this overloaded syntax, we may often read such an entity $e$ (implicitly) as its lifting; and exception is states which relate differently as indicated by the definition of the covering function. 

To characterize sets of object terms, formulas, etc.\ with a single meta-level term, we assume a collection of logic \emph{meta-level constraints} whose meanings are given by a first-order theory $\Mmetaprolog$.
\begin{definition}\label{def:meta-prolog-theory} 
The theory $\Mmetaprolog$ includes at least the following constraints.
\begin{itemize}
  \item {\rm\texttt=}\emph{/2} with its usual meaning of syntactic identity,
  \item \emph{Type constraints} {\rm\texttt{type}}\emph{/2}. For example {\rm\texttt{type(var,$x$)}} is  true in $\Mmetaprolog$ whenever $x$ is the name of an object level variable; {\rm\texttt{var}} is an example of a \emph{type}, and we introduce more types below when we need them. 
  \item \emph{Modal constraints} $\willsucced F$, $\willfail F$, and $\willerror F$
     are defined to be true in $\Mmetaprolog$ whenever $\exe(\denotesgr F)$ returns a proper substitution, a failure substitution and an error substitution, respectively.
  \item The \emph{extended modal constraint} $\willsucced_H F$ is true in $\Mmetaprolog$ whenever $\willsucced F$ is true and $\denotesgr{H}=\denotesgr{H} \exe(\denotesgr F)$ holds (the substitution does not bind variables in $H$).
  \item {\rm\texttt{freshVars($L$,$T$)}}  is true in $\Mmetaprolog$ whenever $L$ is a list of all different variables names, none of which occur in the term $T$; {\rm\texttt{freshVars($L_1$,$L_2$,$T$)}} ab\-bre\-vi\-a\-tes
  {\rm\texttt{freshVars($L_{12}$,$T$)}} where $L_{12}$ is the concatenation of $L_1$ and $L_2$. 
\end{itemize}
\end{definition}
Substitutions are defined as usual at the meta-level, and we denote the set of grounding substitutions that satisfy a 
meta-level constraints as by 
$
[M] = \{ \sigma \mid \Mmetaprolog\models M\sigma\}.
$
A \emph{constrained meta-level term} is of the form
$mt\where M$
where $mt$ is a meta-level term and $M$ a conjunction of meta-level constraints. Such constrained meta-level terms may describe infinite sets of object-level entities.
The \emph{covering function}
$\denotesmetaprolog-$ 
from constrained meta-level terms into sets of object level notions is defined as 
$\denotesmetaprolog{mt \where M}=\ \{\denotes{mt}^\sigma\mid \sigma\in [M]\}$
where the notation $\denotes{-}^\sigma$ is defined as follows.
\begin{eqnarray*}
\denotes{\langle S,B\rangle}^\sigma&\quad=\quad & \denotes{S}^\sigma \exe(\denotes{B}^\sigma)\\
\denotes{mt}^\sigma&\quad=\quad& \denotesgr{mt\,\sigma},\quad\mbox{when $mt$ does not contain $\langle-,-\rangle$}\\
\denotes{f(mt_1,\ldots,mt_n)}^\sigma &\quad=\quad& f(\denotes{mt_1}^\sigma,\ldots,\denotes{mt_n}^\sigma),\quad\mbox{otherwise}.
\end{eqnarray*}
A constrained meta-level term $T$ is \emph{inconsistent} whenever $\denotesmetaprolog{T}=\emptyset$;
otherwise, it is \emph{consistent}.

\subsection{Meta-level states}
A \emph{meta-level state representation} is a constrained meta-level term $\Sigma$ 
which $\denotesmetaprolog{\Sigma}$ is a set of object level state representations.
Two meta-level state representations $\Sigma_1$ and $\Sigma_2$ are \emph{variants} whenever
$\denotesmetaprolog{\Sigma_1}=\denotesmetaprolog{\Sigma_2}$.
A \emph{meta-level state} is an equivalence class of meta-level state representations under the variant relationship.
Analogous to the object level, we typically indicate a given meta-level state by one of its
meta-level state representations.
 \begin{example}\label{ex:coveredM}
Consider the following meta-level state representation.
$$SR_1 \;\; = \;\; \bigl(\langle \{\texttt{p($x$)}\},\texttt{$x$=$c$}\rangle \where \texttt{type(var,$x$)}{\land} \texttt{type(const,$c$)}\bigr)
$$
It covers all CHR states $ \{\texttt{p($c$)}\}$ in which $c$ is a constant.
We will use below, the property that we can add consistent meta-level constraints on variables not mentioned elsewhere
without changing the set of covered CHR states. For example, $\denotes{SR_1}=\denotes{SR_2}$. i.e.,
$SR_1$ and $SR_2$ are variants, where
$$SR_2 \;\; = \;\; \bigl(\langle \{\texttt{p($x$)}\},\texttt{$x$=$c$}\rangle \where \texttt{type(var,$x$)}{\land} \texttt{type(const,$c$)}{\land} \texttt{type(int,$i$)}\bigr).
$$
\end{example}
A consistent meta-level state $\Sigma$ is called 
   \emph{proper} whenever every object level state in $\denotesmetaprolog\Sigma$ is proper, 
  a \emph{failed} (an \emph{error}) state whenever $\denotesmetaprolog\Sigma=\{\failure\}$ ($=\{\error\})$, and
  \emph{mixed} whenever  $\denotesmetaprolog\Sigma$ contains two or more of 
  $\failure$, $\error$
  or a proper state.  
  
A consistent state repr.\ $(\langle S,B \rangle\where M)$ is \emph{reduced} if $B=\true$,
and it is \emph{reducible to} a state repr.\ $(\langle S',\true\rangle\where M)$ 
whenever $\denotesmetaprolog{(\langle S,B \rangle, \;\langle S',\true\rangle) \break \where M}$ is the identity relation.
Notice that this property is stronger than variance introduced above,
as the meta-level constraints are the same in both state repr.s and common meta-variables in the two
are instantiated simultaneously. 

\subsection{Operational semantics}
A \emph{meta-level rule application} wrt.\ a proper and reduced meta-level state $\langle H_1\uplus H_2 \uplus S, \true \rangle \where M$ is a lifted pre-application 
$\nonamerule {H_1}{H_2}GC$ with local variables $L$ and head variables $H$, such that $$\Mmetaprolog\;\models\; M\;\rightarrow\; \willsucced_H G \land \texttt{freshVars($L$,$H_1\uplus H_2 \uplus S$)}.$$

There are two sorts of transitions from proper and reduced meta-level state representations:  
\emph{By rule application}
$$ \Sigma \;\;\;\;\ourlongmapstometaprolog\;\;\;\; \langle H_1\uplus C \uplus S, G\rangle \where M $$
whenever $\Sigma$ is reducible to $\langle H_1\uplus H_2 \uplus S,\true\rangle\where M$, whenever there exists a meta-level rule application 
$\nonamerule {H_1}{H_2}GC$. 
 \emph{By built-in}
 $$ \Sigma \;\;\;\;\ourlongmapstometaprolog\;\;\;\; \langle S,b\rangle\where M $$
whenever $\Sigma$ is reducible to $\langle S\uplus \{b\},\true\rangle\where M.$
Notice for a built-in transition, that the resulting state may be  mixed.
\begin{example}
Consider the rule $\namedsimprule{r_3}{\texttt{q(X)}}{\texttt{X>0}}{\texttt{r(X)}}$ extracted from Example~\ref{ex:zigzagPlain}
below.
Read as a meta-level application instance, it gives rise to the following transition.
$$
\bigl(\langle\{\texttt{q($x$)},\texttt{s($x$)}\},\true\rangle\ourlongmapstometaprolog \langle\{\texttt{r($x$)},\texttt{s($x$)}\},\texttt{$x$>0}\rangle\bigr)
\where \willsucced_{x} (\texttt{$x$>0})
$$
The \texttt{freshVars}/2 
constraint is not needed as there are no local variables.
\end{example}

\begin{lemma}\label{lm:cover2}
Let $\Sigma$ be a proper and reducible s.repr. 
If there is a meta-level transition $\Sigma \Mapsto \Sigma'$ then any element $S \rightarrowsubprolog S'$ in $\denotesmetaprolog{\Sigma \Mapsto \Sigma'}$ is a CHR transition. 
\end{lemma}
\begin{proof} 
Let $(\langle A,B\rangle
 \;\Mapsto\;   \langle A', B'\rangle) \where M$ be a meta-level transition and $\sigma$ an arbitrary substitution in $[M]$. We have to show that 
 $(\denotes{\langle A,B\rangle}^\sigma
 \rightarrowsubprolog  \denotes{\langle A', B'\rangle}^\sigma)$ =
 $  
 (\denotesgr{A\sigma}\exe(\denotesgr{B\sigma})
 \;\rightarrowsubprolog\; \denotesgr{(A'\sigma)}\exe(\denotesgr{B'\sigma}))$
 is a transition in the object system. 
 We consider the two cases, by rule application and by built-in.  

\emph{By rule application}. By def.\ of meta-level transition,  $(\langle A,B\rangle \where M)$ is proper and reducible to a form $\left( \langle H_1\uplus H_2 \uplus S,\true\rangle\where M \right)$, obtaining that
$\denotesgr{H_1\sigma\uplus H_2\sigma \uplus S\sigma}=\denotesgr{A\sigma} \exe(\denotesgr{B\sigma})$ is proper and $\denotesgr{(A'\sigma)}\exe(\denotesgr{B'\sigma})$
$= \denotesgr{H_1\sigma\uplus C \uplus S\sigma}\exe(\denotesgr{G\sigma})$. There exists a meta-level rule application, i.e., 
 a lifted pre-application $\nonamerule {H_1}{H_2}GC$ with local variables $L$ and head variables $H$, such that $\Mmetaprolog\models M\rightarrow \willsucced_H G \land \texttt{freshVars($L$,$H_1\uplus H_2 \uplus S$)}$. Thus $\exe(\denotesgr{G\sigma})$ is a proper substitution that does not instantiate head variables $H$, and since these conditions are a direct formalization of the requirements for object level transition, $\denotesgr{H_1\sigma \text{\rm\texttt{\char92}}H_2\sigma\;\text{\rm\texttt{<=>}} \;G\sigma \text{\rm\texttt{|}} C\sigma}$ is an object level application instance, which leads to the conclusion. 

 \emph{By built-in}. 
 By def.\ of meta-level transition,  $(\langle A,B\rangle \where M)$ is proper and reducible to a form $\left(S\uplus\{b\},\true\rangle\where M \right)$ and 
 $\langle A', B'\rangle$ is the same as $\langle S, b\rangle$.
By definition of ``reducible'', $\denotesgr{A\sigma} \exe(\denotesgr{B\sigma}) = \denotesgr{S\sigma \uplus \{b\sigma\}}=\denotesgr{ S\sigma} \uplus \denotesgr{b\sigma}$ which is then a proper state, thus,
 $(\denotesgr{ S\sigma} \uplus \denotesgr{b\sigma}) \rightarrowsubprolog   \denotesgr{S\sigma}\exe(\denotesgr{b\sigma})=\denotesgr{A'\sigma}\exe(\denotesgr{B\sigma})$ concludes the proof.
\qed\end{proof}

\subsection{Invariants and Equivalences}
\emph{Meta-level invariants} $\Imetaprolog$ and \emph{equivalences} $\approxmetaprolog$ are defined as follows.
\begin{itemize}
  \item $\Imetaprolog(S)$ whenever $I\subprolog(s)$  for all $s\in\denotesmetaprolog S$.
  \item $S_1 \approxmetaprolog S_2$ whenever $s_1 \approxprolog s_2$  for all $(s_1,s_2)\in \denotesmetaprolog{(S_1,S_2)}$. 
\end{itemize}
The following properties that $\Imetaprolog$ is in fact an invariant and $\approxmetaprolog$ is  in fact an equivalence relation are direct consequences of the above definition and Lemma~\ref{lm:cover2}.
\begin{proposition}\label{prop:invmeta}
If $\Imetaprolog(\Sigma_1)$ and $\Sigma_1  \Mapsto \Sigma_2$ then $\Imetaprolog(\Sigma_2)$.
\end{proposition}
\begin{proposition}\label{prop:eq}
The $\approxmetaprolog$ relation is an equivalence relation, i.e., reflexive, transitive and symmetric.
\end{proposition}
To express meta-level invariants and equivalences we assume two meta-level constraints.
\begin{definition} The theory $\Mmetaprolog$ includes 
 two constraints $\cinv$ and $\cequiv$
  such that $\cinv(\Sigma)$ is true in $\Mmetaprolog$ whenever $\denotesgr{\Sigma}$ is an $\Iprolog$ state (representation) of the Prolog-based semantics,
  and $\cequiv(\Sigma_1,\Sigma_2)$ whenever  $\denotesgr{(\Sigma_1,\Sigma_2)}$ is a pair of
  state  (representation)s $(s_1,s_2)$  such that  $s_1\approxprolog s_2$.
\end{definition}
The following correspondences follow immediately from the definitions above.
\begin{proposition}\label{prop:lift-inv-and-equiv}
$\Imetaprolog  \;=\;  \{S \where M \mid  \Mmetaprolog\models M\rightarrow \cinv(S) \}$\\
\hbox to 7.5em{}$\approxmetaprolog  \;=\;  \{ (S_1,S_2)\where M \mid \Mmetaprolog\models M \rightarrow \cequiv(S_1, S_2) \}$
\end{proposition}
\begin{example}[Example~\ref{ex:zigzagPlain} cont']\label{ex:zigzagPlain-inv}
The invariant defined as ``reachable from an initial state \texttt{p($n$)} where $n$ is a number'' for the program in Example~\ref{ex:zigzagPlain} may be formalized at the meta-level as states of the form $\langle\{\mathit{pred}\texttt{($n$)}\},B\rangle\where \texttt{type(num,$n$)} \land \willsucced_n B$ where $\mathit{pred}$ is one of
\texttt{p}, \texttt{q} and \texttt{r}, and $B$ is one of \texttt{$n\ttleq$0}, $n\texttt{>0}$, and $\true$. 
\end{example}
\begin{example}\label{ex:set-inv-eq} (Example~\ref{ex:set} continued)
Consider again the \texttt{set}-program. An invariant relevant for this program may be formalized
 as states of the form
{\small$$\langle \{\texttt{set($L$)}\}{\uplus}S,\true\rangle \where \texttt{type(constList,$L$)}{\land}\texttt{type(constItems,$S$)};$$}\noindent
we assume types \texttt{const} for all constants,
\texttt{constList} for all lists of such, and \texttt{constItems} for sets of constraints of the form \texttt{item($c$)}
where $c$ is a constant.

A state equivalence relevant for this program may be formalized as follows
{\small$$\langle \{\texttt{set($L_1$)}\}{\uplus}S,\true\rangle \where M^{\approx}
\quad\approxmetaprolog\quad
\langle \{\texttt{set($L_2$)}\}{\uplus}S,\true\rangle \where M^{\approx}
$$}\noindent
where $M^{\approx}$ stands for {\small$\texttt{type(constList,$L_1$)}{\land}
\texttt{type(constList,$L_1$)}{\land}\texttt{perm($L_1$,$L_2$)}
{\land}\break\texttt{type(constItems,$S$)}$}, and 
 {\small\texttt{perm($L_1$,$L_2$)}} means
that  $L_1$ is a permutation of $L_2$.

\end{example}

\subsection{Meta-level joinability and splitting}
Joinability at the meta-level is related to joinability at the object level.
\begin{proposition}\label{prop:coverjoinable}
If a meta-level corner is joinable modulo $\approxmetaprolog$ then all corners covered by it are joinable modulo $\approxprolog$.
\end{proposition}
\begin{proof}
This is a direct consequence of Lemma~\ref{lm:cover2} and definition of $\approxmetaprolog$.
\end{proof}
However, a meta-level corner covering only joinable (mod.\ $\approxprolog$) object corners may not be joinable (mod.\ $\approxmetaprolog$), as demonstrated by following example.
\begin{example}(Example~\ref{ex:zigzagPlain}-\ref{ex:zigzagPlain-inv} cont')\label{ex:zigzagPlain-cov}
Consider the meta-level corner
$$\small \Lambda = \langle \{\texttt{q($n$)} \},\true\rangle \Mapsfrom \langle\{\texttt{p($n$)} \},\true\rangle  \Mapsto  
\langle\{\texttt{r($n$)} \},\true\rangle \qquad  \where \texttt{type(num,$n$)}.$$
No rule apply to either of the wing states and, thus, they are not joinable. However, all the covered corners are joinable, e.g.\ using rule $r_3$ for $\texttt{q(1)} \ourmapsto  \texttt{r(1)}$ and rule $r_4$ for $\texttt{q(-1)} \ourmapsfrom \texttt{r(-1)}$.
\end{example}
To accommodate this phenomenon we introduce splitting of meta-level terms such as states, transitions, equivalences, or corners and afterwards split-joinability (mod.\ $\approxmetaprolog$).

Let $(mt \where M)$ be a constrained meta-level term, and
$\{M_i\}_{i\in\mathit{Inx}}$  a set of meta-level constraints such that
$\Mmetaprolog\models \left(\bigvee_{i\in\mathit{Inx}} M_i\right) \leftrightarrow \true.$ 
We can now define a \emph{splitting} of $(mt \where M)$ into the following set of constrained meta-level terms
$$\{mt \where M\land M_i \}_{i\in\mathit{Inx}}.$$
The following is a direct consequence of the definitions.
\begin{proposition}
If $\{mt \where M\land M_i \}_{i\in\mathit{Inx}}$ is a splitting of $(mt \where M)$ then $\denotes{mt \where M} = \bigcup_{i\in\mathit{Inx}}\denotes{mt \where M\land M_i}$.
\end{proposition}
\begin{example}(Example~\ref{ex:zigzagPlain}, \ref{ex:zigzagPlain-inv}, \ref{ex:zigzagPlain-cov} cont')\label{ex:zigzag-split}
The meta-level corner $\Lambda$ (Example~\ref{ex:zigzagPlain-cov}) is a constrained meta-level term and $\{\willsucced(n\texttt{>0}), \willsucced(n \ttleq 0)\}$ is a set of constraints with $\Mmetaprolog\models (\willsucced(n\texttt{>0}) 
\lor \willsucced(n \ttleq 0) )
 \leftrightarrow \true.$
Thus, we may split $\Lambda$ into the corners
$$\small \langle \{\texttt{q($n$)} \},\true\rangle \Mapsfrom \langle\{\texttt{p($n$)} \},\true\rangle  \Mapsto  
\langle\{\texttt{r($n$)} \},\true\rangle \quad  \where \texttt{type(num,$n$)}\land \willsucced(n\texttt{>0})$$ 
$$\small \langle \{\texttt{q($n$)} \},\true\rangle \Mapsfrom \langle\{\texttt{p($n$)} \},\true\rangle  \Mapsto  
\langle\{\texttt{r($n$)} \},\true\rangle \quad  \where \texttt{type(num,$n$)}\land \willsucced(n \ttleq 0) .$$
\end{example}
\begin{definition}[split-joinable]\label{def:splitjoin}
A meta-level corner $\Lambda$  is split-joinable modulo $\approxmetaprolog$  if there is a splitting $\{\Lambda_i\}_{i\in\mathit{Inx}}$ such that each
$\Lambda_i$ is joinable modulo $\approxmetaprolog$.
\end{definition}
\begin{example}(Example~\ref{ex:zigzagPlain}, \ref{ex:zigzagPlain-inv}, \ref{ex:zigzagPlain-cov} cont')\label{ex:zigzag-splitjoin}
The splitted corners of Example~\ref{ex:zigzag-split} are joinable as follows; thus, the corner $\Lambda$ of Example~\ref{ex:zigzagPlain-cov} is split-joinable.  
$$\small \langle \{\texttt{q($n$)} \},\true\rangle \Mapsto 
\langle\{\texttt{r($n$)} \},\true\rangle \quad  \where \texttt{type(num,$n$)}\land \willsucced(n\texttt{>0})$$ 
$$\small \langle \{\texttt{q($n$)} \},\true\rangle \Mapsfrom  
\langle\{\texttt{r($n$)} \},\true\rangle \quad  \where \texttt{type(num,$n$)}\land \willsucced(n \ttleq 0) .$$ 
\end{example}

\begin{proposition}\label{prop:meta-corner}
A meta-level corner is split-joinable modulo $\approxmetaprolog$ if and only if all corners covered by it are joinable modulo $\approxprolog$.
\end{proposition}
\begin{proof}
``\textit{If}'':
Let $\Lambda$ denote a meta-level corner that is split-joinable modulo $\approxmetaprolog$; that is there exists an indexed set of corners 
$\{\Lambda_i\}_{i\in \mathit{Inx}}$
such that
$\denotes{\Lambda} = \bigcup_{i \in \mathit{Inx}}  \denotes{\Lambda_i}$ where each $\Lambda_i$ is joinable modulo $\approxmetaprolog$.
Thus, for any object corner $\lambda$  covered by $\Lambda$, there is some $\Lambda_i$ that covers $\lambda$ and is joinable modulo $\approxmetaprolog$. By Proposition~\ref{prop:coverjoinable}, $\lambda$ is joinable modulo $\approxprolog$.

``\textit{Only if}'':
Index the set of object corners by $\mathit{Inx}$, and lift every object corner $\lambda_i$ to a ground meta-level corner $\Lambda_i$ such that $\denotes{\Lambda_i}=\{\lambda_i\}$. Now, we choose a splitting of the meta-level corner $\Lambda$ into a set $\cup_{i\in \mathit{Inx}} \Lambda_i$, each of which is joinable modulo $\approxmetaprolog$, obtaining, by Definition~\ref{def:splitjoin}, that  $\Lambda$ is split-joinable modulo $\approxmetaprolog$. \qed
\end{proof}

\section{Confluence}\label{sec_confluence}
As noted in Section~\ref{sec:prelim}, (obs.) confluence (mod.\ equiv.)
of a terminating system can be shown by (obs.) local confluence (mod.\ equiv.).  Similarly to previous work we define a smaller set of critical corners whose joinability ensures this property.
\begin{definition}[critical $\alpha$-corners]\label{def:crit-alphacorner}~\\ 
\textbf{critical $\alpha_1$-corners.}
Let $\langle A \uplus S, \true \rangle \where M$ be a proper and reduced meta-level state where $A=(H_1 \uplus H_2) \cup (H'_1 \uplus H'_2)$, $M = \willsucced_H G \land  \willsucced_{H'} G' \land  \cfreshvars(L,L',A\uplus S)$ and $(H_1 \uplus H_2) \cap H'_2 \neq \emptyset$. Furthermore, let  
$\nonamerule {H_1}{H_2}GC$ and
$\nonamerule {H'_1}{H'_2}{G'}{C'}$ be two  meta-level rule applications with local and head variables $L$, $H$, and $L'$, $H'$, respectively.
A \emph{critical $\alpha_1$-corner} is defined as
$$\langle (A\! \setminus\! H_2) \uplus C \uplus S, G \rangle
 \;\ourlongmapsfrommetaprolog \;
 \langle A \uplus S, \true \rangle
 \;\ourlongmapstometaprolog \;
\langle (A\!\setminus\! H'_2) \uplus C' \uplus S, G' \rangle \qquad \where M.$$ 
\textbf{critical $\alpha_2$-corners.} 
Let $\langle A \uplus S, \true \rangle \where M$ be a meta-level state where $A=H_1 \uplus H_2\uplus \{b\}$, $ M = \willsucced_H G \land \cfreshvars(L,A \uplus S)$, $b$ is a built-in template and there is a meta-level rule application $\nonamerule {H_1}{H_2}GC$  with local and head variables $L$ and $H$.  
A \emph{critical $\alpha_2$-corner} is defined as\\
$$\;\;\langle (A\!\setminus\! H_2) \uplus C \uplus \{b\} \uplus S, G \rangle
 \ourlongmapsfrommetaprolog 
 \langle A \uplus S, \true \rangle
 \ourlongmapstometaprolog 
\langle A\!\setminus\!\{b\} \uplus S, b \rangle \qquad \where M.$$ 
\textbf{critical $\alpha_3$-corners.}
Let $\langle \{b, b'\}  \uplus S, \true \rangle \where M$ be a meta-level state where $b$ and $b'$ are built-in templates. 
A \emph{critical $\alpha_3$-corner} is defined as\\
$$ \langle \{b'\} \uplus S, b \rangle 
 \;\ourlongmapsfrommetaprolog\; 
 \langle \{b, b'\}  \uplus S, \true \rangle
 \;\ourlongmapstometaprolog \;
\langle \{b\} \uplus S, b' \rangle \qquad \where M.$$
\end{definition}
\begin{example}(Example~\ref{ex:empty-is-non-confl} cont')\label{ex:empty-is-non-confl-meta}
Assuming built-ins \texttt{is}/2 and \texttt{=}/2, the empty program is shown non-confluent; the critical $\alpha_3$-corner $( \error \,\Mapsfrom \, \{ x \; \texttt{is}\, y\texttt{+1}, y\texttt{=2} \} \,\Mapsto \{x \,  \texttt{is}\, \texttt{2+1}\})  \where \texttt{type(var,$x$)} \land \texttt{type(var,$y$)}$ is not split-joinable. 
The program is confl.\ under invariants excluding built-ins; the corners are inconsistent.
\end{example}
\begin{example}(Example~\ref{ex:zigzagPlain}, \ref{ex:zigzagPlain-inv}, \ref{ex:zigzagPlain-cov}-\ref{ex:zigzag-splitjoin} cont')\label{ex:zigzagPlain-corners}
Consider again the non-confluent program without the invariant. There is a single non-split-joinable critical $\alpha_1$-corner $ \texttt{q($n$)} \Mapsfrom  \texttt{p($n$)} \Mapsto  \texttt{r($n$)}\where \true$. There is a series of joinable $\alpha_2$-corners where the rule and the built-in transition commute. There are many non-joinable critical $\alpha_3$-corners, see Example~\ref{ex:empty-is-non-confl-meta}. 
\end{example}
\begin{definition}\label{def:crit-obs-alphacorner}
Let $(\Sigma_1 \Mapsfrom  \Sigma_0 \Mapsto \Sigma_2 ) \where M$ be a critical $\alpha_i$-corner for $i\in\{1,2,3\}$;
a \emph{critical observable $\alpha_i$-corner} is defined as $(\Sigma_1 \Mapsfrom  \Sigma_0 \Mapsto  \Sigma_2 ) \where M \land \cinv(\Sigma_0)$.
\end{definition}
Notice that by Proposition~\ref{prop:invmeta}, it follows that $\Mmetaprolog\models M\rightarrow \cinv(\Sigma_1) \land \cinv(\Sigma_2)$.

\begin{theorem}
An observable terminating program $\Pi$ is observably confluent iff its observable critical $\alpha$-corners are joinable.
\end{theorem}
\begin{proof}
A direct consequence of the subsequent Theorem~\ref{thm:crit-corner}.
\end{proof}
\begin{example}(Example~\ref{ex:zigzagPlain}, \ref{ex:zigzagPlain-inv}, \ref{ex:zigzagPlain-cov}-\ref{ex:zigzag-splitjoin}, \ref{ex:zigzagPlain-corners}, cont')\label{ex:zigzagPlain-obscorners}
Consider the program and its invariant. There is a single critical $\alpha_1$-corner, namely 
 $\Lambda$ of Example~\ref{ex:zigzagPlain-cov} that was shown split-joinable. There are zero observable critical $\alpha_2$-corners and $\alpha_3$-corners. Thus, the program is observably confluent.
\end{example}

\begin{definition}[critical $\beta$-corners]
\textbf{critical $\beta_1$-corners.}
Let $(\langle A, \true \rangle \where M)$ be a consistent meta-level state, $(\langle H_1 \uplus H_2 \uplus S, \true \rangle \where M)$ a proper meta-level state with $M{=}\willsucced_H G \land \cfreshvars(L,H_1 \uplus H_2 \uplus S) \land \cequiv(\langle A, \true \rangle,  \langle H_1 \uplus H_2 \uplus S, \true \rangle)$, and $\nonamerule {H_1}{H_2}GC$ a meta-level rule application with local variables $L$ and head variables  $H$.  
A critical $\beta_1$-corner is defined as
$$\langle A, \true \rangle 
\approxmetaprolog 
\langle H_1 \uplus H_2 \uplus S, \true \rangle 
\;\ourlongmapstometaprolog\;
\langle H_1 \uplus C \uplus S, G \rangle 
\qquad \where M.$$
\textbf{critical $\beta_2$-corners.}
Let $(\langle A, \true \rangle \where M)$ and $(\langle \{b\} \uplus S, \true \rangle \where M)$ be meta-level states where 
$M= \cfreshvars(L,H_1 \uplus H_2 \uplus S) \land \cequiv(\langle A, \true \rangle, \break \langle H_1 \uplus H_2 \uplus S, \true \rangle)$ and
$b$ is a built-in template.
A critical $\beta_2$-corner is defined as
$$\langle A, \true \rangle 
\approxmetaprolog 
\langle \{b\} \uplus S, \true \rangle 
\;\ourlongmapstometaprolog\;
\langle S, b \rangle 
\qquad \where M.$$
\end{definition}
\begin{definition}
Let $(\Sigma_1 \approxmetaprolog \Sigma_0 \Mapsto \Sigma_2) \where M$ be a critical $\beta_i$-corner for $i\in\{1,2\}$.
A \emph{critical observable $\beta_i$-corner} is defined as $(\Sigma_1 \approxmetaprolog  \Sigma_0 \Mapsto  \Sigma_2 ) \where M \land \cinv(\Sigma_0) \land \cinv(\Sigma_1)$.
\end{definition}
Note that, by Proposition~\ref{prop:invmeta}, it follows that $\cinv(\Sigma_2)$.

\begin{lemma}[Critical Corner Lemma]\label{lm:crit-corner}
Any observable object level corner $s_1 \approxprolog s_0 \rightarrowsubprolog s_2$ (resp.\  $s_1 \leftarrowsubprolog s_0 \rightarrowsubprolog s_2$) is either joinable modulo $\approxprolog$, or it is covered by an observable critical $\alpha$-corner (resp.\  $\beta$-corner).
\end{lemma}
\begin{proof}
We will go through the different sorts of observable corners and show that either they are not covered by a corner and joinable, or they are covered by some observable critical corner.
In the following we let $\stackrel{r}{\mapsto}$ refer to a transition by a rule
application being an instance of rule $r$, and $\stackrel{b}{\mapsto}$ refer to
a transition by built-in.
\begin{description}
\item[$ s_1 \stackrel{r}{\leftarrowsubprolog} s_0 \stackrel{r'}{\rightarrowsubprolog} s_2$]
This is possible in three ways. 
\begin{inparaenum}[(1)] \item There is no overlap between the heads of the rule applications, thus, the rule applications commute. 
\item The rule applications $r= \nonamerule{h_1}{h_2}{g}{c}$ and $r'= \nonamerule{h'_1}{h'_2}{g'}{c'}$ have an overlap $(h_1 \uplus h_2) \cap h'_2 \neq \emptyset$. Hence, we can construct an observable critical $\alpha_1$-corner $\Sigma_1 \Mapsfrom \Sigma_0 \Mapsto \Sigma_2 \where M$ such that $(s_1 \stackrel{r}{\leftarrowsubprolog} s_0 \stackrel{r'}{\rightarrowsubprolog} s_2) \in \denotes{\Sigma_1 \Mapsfrom \Sigma_0 \Mapsto \Sigma_2 \where M}$.
\item There is no overlap $(H_1 \uplus H_2) \cap H'_2 = \emptyset$ but there is one by $(H'_1 \uplus H'_2) \cap H_2 \neq \emptyset$; in this case Definition~\ref{def:crit-alphacorner} states that there exists an observable critical $\alpha_1$-corner $\Sigma_1 \stackrel{r'}{\leftarrowsubprolog} \Sigma_0 \stackrel{r}{\rightarrowsubprolog} \Sigma_2$ covering the corner
$s_2 \stackrel{r'}{\leftarrowsubprolog} s_0 \stackrel{r}{\rightarrowsubprolog} s_1$, thus, by symmetry it also covers this corner.  
\end{inparaenum} 
\item[$ s_1 \stackrel{r}{\mapsfrom} s_0 \stackrel{b}{\mapsto} s_2$] 
  Refer to the corner in question as $\lambda$, and let $\Lambda$ be a the critical meta-level corner constructed from
  $r$ (Definition~\ref{def:crit-obs-alphacorner}) and the built-in template of the built-in predicate in $b$.
  It is straightforward to show that $\lambda\in\denotesmetaprolog\Lambda$; thus
  there are no such object level $\alpha_2$ corner that is not covered by  some critical $\alpha_2$ meta-level corner.
\item[$(s_1 \stackrel{b}{\mapsfrom} s_0 \stackrel{b}{\mapsto} s_2), (s_1 \approxprolog s_0 \stackrel{r}{\mapsto} s_2), (s_1 \approxprolog s_0 \stackrel{b}{\mapsto} s_2)$] The proofs are similar in structure to the previous and is left out. \qed
\end{description}\end{proof}

\begin{proposition}\label{prop:critjoincoverjoin}
An observable critical corner is split-joinable modulo $\approxmetaprolog$ if and only if every corner covered by it is joinable modulo $\approxprolog$.
\end{proposition}
\begin{proof}
This is a special case of Proposition~\ref{prop:meta-corner}. \qed
\end{proof}

\begin{theorem}[Critical Corner Theorem]\label{thm:crit-corner}
A program is observable locally confluent modulo $\approxprolog$ if and only if its observable critical $\alpha$- and $\beta$-corners are split-joinable modulo $\approxmetaprolog$.
\end{theorem}
\begin{proof}
$\Leftarrow$: Assume that the observable critical $\alpha$- and $\beta$-corners  are split-joinable modulo $\approxmetaprolog$; by the Critical Corner Lemma (Lemma~\ref{lm:crit-corner}) and Proposition~\ref{prop:critjoincoverjoin} we conclude that the program is observable locally confluent. $\Rightarrow$:
Assume that the program is observable locally confluent modulo $\approxprolog$, thus
all corners covered by observable critical corners are joinable modulo $\approxprolog$, thus, by Proposition~\ref{prop:critjoincoverjoin}, the observable critical corners are also split-joinable modulo $\approxmetaprolog$. \qed
\end{proof}
\begin{theorem}
An observable terminating CHR program is observably confluent modulo $\approxprolog$ if and only if its observable critical $\alpha$- and $\beta$-corners are split-joinable modulo $\approxmetaprolog$.
\end{theorem}
\begin{proof}
A direct consequence of Theorem~\ref{thm:crit-corner} and Lemma~\ref{lm:huet-mod-eq} generalized under invariant, i.e.,\ an observable terminating system is observably confluent modulo $\approxprolog$ if and only if it is observable locally confluent modulo $\approxprolog$. \qed
\end{proof}
When showing joinability (modulo equivalence) of the observable critical corners we will leave out duplicates and corners with similar wing-states.
\begin{example}[\cite{HenningMajaLOPSTR18}](Example~\ref{ex:set} cont') 
Consider the one-rule \texttt{set}-program  together with the invariant and equivalence from Example~\ref{ex:set-inv-eq}.
There are two critical $\alpha$-corners, given by the two ways, the rule can overlap with itself. \vspace{-0.5cm} \begin{center}
\tikzset{|/.tip={Bar[width=.8ex,round]}}
\noindent\begin{tikzpicture}[|->, auto,node distance=2cm and 1.0cm, 
               main node/.style={font=\normalsize 
              }]
  \node[main node] (1) {$\langle\{\texttt{item($x_1$)},\texttt{set($L_1$)},\texttt{item($x_2$)}\}, \true\rangle$};
  \node[main node] (2) [above =0.1cm of 1] {$\langle \{\texttt{set([$x_1$|$L_1$])},\texttt{item($x_2$)}\}, \true\rangle $};
  \node[main node] (3) [below =0.1cm of 1] {$\langle \{\texttt{item($x_1$)},\texttt{set([$x_2$|$L_1$])}\}, \true\rangle$};
  \path
    (1) edge [|-Implies, line width=0.4pt, double distance=1.5pt, shorten >=-1pt, shorten >=-3pt, shorten <=-2pt]  
(2)
    (1) edge [|-Implies, line width=0.4pt, double distance=1.5pt, shorten >=-1pt, shorten >=-3pt, shorten <=-2pt] 
 (3);
 
\begin{scope}[xshift = 6.2 cm]
  \node[main node] (1) {$\langle\{\texttt{set($L_1$)},\texttt{item($x_1$)},\texttt{set($L_2$)}\}, \true\rangle$};
  \node[main node] (2) [above =0.1cm of 1] {$\langle\{\texttt{set([$x_1$|$L_1$])},\texttt{set($L_2$)}\}, \true\rangle$};
  \node[main node] (3) [below =0.1cm of 1] {$\langle\{\texttt{set($L_1$)},\texttt{set([$x_1$|$L_2$])}\}, \true\rangle$};
  \path
    (1) edge [|-Implies, line width=0.4pt, double distance=1.5pt, shorten >=-1pt, shorten >=-3pt, shorten <=-2pt] 
     (2)
    (1) edge [|-Implies, line width=0.4pt, double distance=1.5pt, shorten >=-1pt, shorten >=-3pt, shorten <=-2pt] 
   (3);
\end{scope}
\end{tikzpicture}
\end{center}\vspace{-0.3cm} 
Each corner is assumed encapsulated in ``$\small\where M$'' with $\small M = 
\texttt{type(const,$x_1$)}\land\texttt{type(const,$x_2$)}\land 
\texttt{type(constList,$L_1$)}{\land}
\texttt{type(constList,$L_2$)}{\land}\texttt{type(constItems,$S$)}$.
The corner on the right is inconsistent due to the two \texttt{set}-constraints and the corner on the left is joinable modulo the equivalence as follows.\vspace{-0.1cm} 
\begin{center}
{\tikzset{|/.tip={Bar[width=.8ex,round]}}

\noindent\begin{tikzpicture}[|->, auto,node distance=2cm and 1.0cm,
               main node/.style={font=             
              }]

  \node[main node] (1) {$\langle\{\texttt{item($x_1$)}, \texttt{set($L$)}, \texttt{item($x_2$)}\}{\uplus}S,\true\rangle$};
  \node[main node] (2) [below left=0.1cm and -3.1cm of 1] {$\langle\{\texttt{set([$x_1$|$L$])}, \texttt{item($x_2$)}\}{\uplus}S,\true\rangle$};
  \node[main node] (3) [below right=0.1cm and -3.1cm of 1] {$\langle\{\texttt{item($x_1$)},\texttt{set([$x_2$|$L$])},\}{\uplus}S,\true\rangle$};
  \node[main node] (2A) [below =0.1cm of 2] {$\langle\{\texttt{set([$x_2$,$x_1$|$L$])}\},\true{\uplus}S\rangle$};
  \node[main node] (3A) [below =0.1cm of 3] {$\langle\{\texttt{set([$x_1$,$x_2$|$L$])}\}{\uplus}S,\true\rangle$};
  \path
    (1) edge [|-Implies, line width=0.4pt, double distance=1.5pt, shorten >=-1pt, shorten >=-3pt, shorten <=-2pt]  (2)
    (2) edge [densely dashed, |-Implies, line width=0.4pt, double distance=1.5pt, shorten >=-1pt, shorten >=-3pt, shorten <=-4pt]  (2A)
    (1) edge [|-Implies, line width=0.4pt, double distance=1.5pt, shorten >=-1pt, shorten >=-3pt, shorten <=-2pt]  (3) 
    (3) edge [densely dashed, |-Implies, line width=0.4pt, double distance=1.5pt, shorten >=-1pt, shorten >=-3pt, shorten <=-4pt]  (3A);
\draw[-, dash pattern=on 2pt off 1pt, transform canvas={yshift=1pt},decorate, decoration={snake, segment length=7pt, amplitude=0.3mm}] (2A) -- (3A);
\draw[-, dash pattern=on 2pt off 1pt, transform canvas={yshift=-1pt},decorate, decoration={snake, segment length=7pt, amplitude=0.3mm}] (2A) -- (3A);
\end{tikzpicture}
}
\end{center}\vspace{-0.3cm}
In addition, there is one $\beta$-corner (assumed encapsulated in ``$\small\where M \land \break \texttt{perm($L_1$,$L_2$)}$'') which is joinable modulo $\approxmetaprolog$ as follows.
\vspace{-0.2cm}\begin{center}
{\tikzset{|/.tip={Bar[width=.8ex,round]}}
\usetikzlibrary{decorations.pathmorphing}
\noindent\begin{tikzpicture}[|->, auto,node distance=2cm and 1.0cm,
               main node/.style={font= ,
               minimum width={5cm}, align=center        
              }]

  \node[main node] (1) {$\langle\{\texttt{item($x$)}, \texttt{set($L_1$)}\} \uplus S,\true\rangle$};
  \node[main node] (2) [below left=0.1cm and -2.5cm of 1] {$\langle\{\texttt{item($x$)}, \texttt{set($L_2$)}\}\uplus S,\true\rangle$};
  \node[main node] (3) [below right=0.1cm and -2.5cm of 1] {$\langle\{\texttt{set([$x$|$L_1$])}\} \uplus S,\true\rangle$};
    \node[main node] (4) [below right=0.1cm and -2.5cm of 2] {$\langle\{\texttt{set([$x$|$L_2$])}\} \uplus S,\true\rangle$};

\draw[-,  transform canvas={yshift=1pt,xshift=-0.4pt},decorate, decoration={snake, segment length=7pt, amplitude=0.3mm}] (1) -- (2); 
\draw[-,  transform canvas={yshift=-1pt},decorate, decoration={snake, segment length=7pt, amplitude=0.3mm}] (1) -- (2);     
   \path   
    (2) edge [densely dashed,|-Implies, line width=0.4pt, double distance=1.5pt, shorten >=-1pt, shorten >=-3pt, shorten <=-2pt]  (4)
    (1) edge [|-Implies, line width=0.4pt, double distance=1.5pt, shorten >=-1pt, shorten >=-3pt, shorten <=-2pt]  (3) ;

\draw[-, dash pattern=on 2pt off 1pt, transform canvas={yshift=1pt,xshift=-0.25pt},decorate, decoration={snake, segment length=7pt, amplitude=0.3mm}] (3) -- (4); 
\draw[-, dash pattern=on 2pt off 1pt, transform canvas={yshift=-1pt},decorate, decoration={snake, segment length=7pt, amplitude=0.3mm}] (3) -- (4); 
 
\end{tikzpicture}
}\end{center}\vspace{-0.3cm}
\noindent 
Thus this proves the program observably locally confluent.
\end{example}

\begin{example}
Consider the \texttt{gcd}-program together with the invariant ``states  with \texttt{gcd}-constraints with positive integers''; we assume a type \texttt{posGcd} for sets of constraints of the form \texttt{gcd($n$)}
where $n$ is a positive integer. 
In the following we indicate the five corners (leaving out one duplicate),
with the proof of joinability hinted for the first one.
\begin{center} \vspace{-0.4cm}
{\tikzset{|/.tip={Bar[width=.8ex,round]}}
\noindent\begin{tikzpicture}[|->, auto,node distance=2cm and 1.0cm,
               main node/.style={font=
              }]
  \node[main node] (1) {$\langle\{\texttt{gcd($n$)}, \texttt{gcd($n$)}, \texttt{gcd($m$)}\}{\uplus}S,\true\rangle$};
  \node[main node] (2) [below left= 0.2 cm and -3.0 cm of 1]
    {$\langle\{\texttt{gcd($n$)}, \texttt{gcd($m$)}\}{\uplus}S,\true\rangle  $};
  \node[main node] (3) [below right= 0.2 cm and -3.0 cm of 1]
  {$\langle\{\texttt{gcd($n$)}, \texttt{gcd($n$)}, \texttt{gcd($l$)}\}{\uplus}S,(n\texttt{<}m, l \,\texttt{is}\, m\texttt{-}n)\rangle$};
  \node[main node] (1A) [below right= 0.2cm and -4.65 cm of 2]
    {$\langle\!\{\texttt{\scalebox{.9}[1.0]{gcd}($n$)}, \texttt{\scalebox{.9}[1.0]{gcd}($l'$)}\}{\uplus}S,(n\texttt{<} m, l' \texttt{is}\, m\texttt{-}n)\!\rangle $};    
\node[main node] (1B) 
[below left= 0.2 cm and -7.8 cm of 3]
    {$\langle\! \{\texttt{\scalebox{.9}[1.0]{gcd}($n$)}, \texttt{\scalebox{.9}[1.0]{gcd}($l$)}\}{\uplus}S,(n\texttt{<} m, l \,\texttt{is}\, m\texttt{-}n)\!\rangle$};      
  \path
    (1.190) edge [|-Implies, line width=0.4pt, double distance=1.5pt, shorten >=-1pt, shorten >=-3pt, shorten <=-2pt] node[left=.2cm ] {$\scriptstyle {r_1}$}
    (2)
    (1.350) edge [|-Implies, line width=0.4pt, double distance=1.5pt, shorten >=-1pt, shorten >= -3pt, shorten <=-2pt] node[right=.3cm ] {$\scriptstyle {r_2}$} 
    (3.170.5);
      \path
    (2) edge [|-Implies, densely dashed, line width=0.4pt, double distance=1.5pt, shorten >=-1pt, shorten >=-3pt, shorten <=-2pt] 
    node[left=.3cm ] {$\scriptstyle {r_2}$} 
    (1A)
    (3.349) edge [|-Implies, densely dashed, line width=0.4pt, double distance=1.5pt, shorten >=-1pt, shorten >= -3pt, shorten <=-2pt] 
    node[right=.3cm ] {$\scriptstyle {r_1}$}  
    (1B);
   
    \draw[-, dash pattern=on 2pt off 1pt, transform canvas={yshift=-1pt},decorate, decoration={snake, segment length=7pt, amplitude=0.3mm}, shorten >= -1pt, shorten <=-1pt] (1A) -- (1B); 
\draw[-, dash pattern=on 2pt off 1pt, transform canvas={yshift=1pt},decorate, decoration={snake, segment length=7pt, amplitude=0.3mm}, shorten >= -1pt, shorten <=-1pt] (1A) -- (1B);
\end{tikzpicture}
}
\end{center}\vspace{-0.2cm}
This corner is assumed encapsulated in
``$\small\where \willsucced_{[n,m]}(n{<} m,l{\,\texttt{is}\,} m\,\texttt{-}\,n) \land \break \cfreshvars(l,\{\texttt{gcd($n$)}, \texttt{gcd($m$)}\} \uplus S) \land    
\texttt{type(posGcd,$\{\texttt{gcd($n$)}, \texttt{gcd($m$)}\} \uplus S$)}$''.
For the proof joinability, i.e., the lower part of the diagram, we need to extend the meta-level constraint with
$\cfreshvars(l',\{\texttt{gcd($n$)}, \texttt{gcd($m$)}\}\uplus S)$, which we can do without 
changing the set of covered CHR states. 
For reasons of space, we indicate the four remaining corners by their common ancestor state; the corners can be determined based on their meta constraints.

{\small
\smallskip\noindent\qquad
$\langle\{\texttt{gcd($n$)}, \texttt{gcd($m$)}, \texttt{gcd($m$)}\}{\uplus}S,\true\rangle\where M
$

\smallskip\noindent\quad\qquad
where $M$ is the same meta-level constraint as in the first corner above.

\smallskip\noindent\qquad
$\langle\{\texttt{gcd($n$)}, \texttt{gcd($m_1$)}, \texttt{gcd($m_2$)}\}{\uplus}S,\true\rangle\where$

\noindent\qquad\qquad$\willsucced_{[n,m_1]}(n{<} m_1,l_1{\,\texttt{is}\,} m_1\,\texttt{-}\,n)\land\willsucced_{[n,m_2]}(n {<} m_2,l_2{\,\texttt{is}\,} m_2\,\texttt{-}\,n)\land$

\noindent\qquad\qquad$\cfreshvars(l_1,l_2, \{\texttt{gcd($n$)}, \texttt{gcd($m_1$)}, \texttt{gcd($m_2$)}\}\uplus S\texttt{)} \land$

\noindent\qquad\qquad$   \texttt{type(posGcd,}\{\texttt{gcd($n$)}, \texttt{gcd($m_1$)}, \texttt{gcd($m_2$)}\}\uplus S\texttt{)}$


\medskip\noindent\qquad
$ \langle\{\texttt{gcd($n_1$)}, \texttt{gcd($n_2$)}, \texttt{gcd($n_3$)}\}{\uplus}S,\true\rangle\where$

\noindent\qquad\qquad$\willsucced_{[n_1,n_2]}(n_1{<} n_2,l_2{\,\texttt{is}\,} n_2\,\texttt{-}\,n_1)\land\willsucced_{[n_2,n_3]}(n_2 {<} n_3,l_3{\,\texttt{is}\,} n_3\,\texttt{-}\,n_2)\land$

\noindent\qquad\qquad$ \cfreshvars(l_2,l_3, \{\texttt{gcd($n_1$)}, \texttt{gcd($n_2$)}, \texttt{gcd($n_3$)}\}\uplus S\texttt{)} \land  $

\noindent\qquad\qquad$ \texttt{type(posGcd,}\{\texttt{gcd($n_1$)}, \texttt{gcd($n_2$)}, \texttt{gcd($n_3$)}\}\uplus S\texttt{)}$


\medskip\noindent\qquad
$ \langle\{\texttt{gcd($n_1$)}, \texttt{gcd($l_1$)}, \texttt{gcd($n_2$)}\}{\uplus}S,(n_1{<} m, l_1 \,\texttt{is}\, m\,\texttt{-}\,n_1)\rangle\where$

\noindent\qquad\qquad$\willsucced_{[n_1,m]}(n_1{<} m,l_1{\,\texttt{is}\,} m\,\texttt{-}\,n_1)\land\willsucced_{[n_2,m]}(n_2 {<} m,l_2{\,\texttt{is}\,} m\,\texttt{-}\,n_2)\land$

\noindent\qquad\qquad$ \cfreshvars(l_1,l_2, \{\texttt{gcd($n_1$)}, \texttt{gcd($n_2$)}, \texttt{gcd($m$)}\}\uplus S\texttt{)} \land  $

\noindent\qquad\qquad$ \texttt{type(posGcd,}\{\texttt{gcd($n_1$)}, \texttt{gcd($n_2$)}, \texttt{gcd($m$)}\}\uplus S\texttt{)}$


}

\end{example}

\section{Conclusion}\label{sec:conclude}
The work presented here is part of a research project aiming to provide
automatic or semi-automatic tools for proving confluence modulo equivalence of CHR programs.
We described a meta-level language
and used it to define a meta-level transition system, able to simulate infinitely many CHR transitions by a single
meta-level transitions.
Our approach is the first to obtain this in the context of invariants and modulo equivalence.
Furthermore, we stepped  
from a theoretically interesting semantics, to one describing CHR as it is implemented and used.
Instead of assuming termination, we proposed the weaker property observable termination and demonstrated it useful. Methods for proving termination of CHR programs have been studied, e.g.,~\cite{DBLP:conf/compulog/Fruhwirth99,DBLP:conf/iclp/PilozziS09a,DBLP:conf/lopstr/PilozziS11}, but they still need to be adapted for observable termination.

A constraint solver for the meta-level language
is under development in CHR as the obvious implementation language. Its main job is to reduce meta-level states or to identify when this is not possible, and to decide satisfiability of a conjunction of meta-level constraints.
Our current experiments indicate that when splitting is not necessary, we can implement proofs of joinability
similarly to existing confluence checkers for CHR, e.g.,~\cite{Raiser-Langbein2010}, referring to the mentioned constraint solver in each step.
When and how to split is an open challenge, and a mechanism is needed for propagating a proposed split backwards through a transition sequence.

\bibliographystyle{abbrv}

\bibliography{CHR}
\end{document}